\documentclass[11pt,leqno]{article}

\usepackage{latexsym}
\usepackage{epsfig}
\usepackage{color}
\usepackage{subfigure}
\usepackage{amssymb,amsmath}

\usepackage[noline,nofillcomment,noend,noresetcount,boxed,figure]{algorithm2e}

\numberwithin{equation}{section}

\newtheorem{theorem}{Theorem}[section]
\newtheorem{lemma}[theorem]{Lemma}
\newtheorem{proposition}[theorem]{Proposition}

\newtheorem{corollary}[theorem]{Corollary}
\newtheorem{example}{Example}[section]
\def\remark #1{\noindent{\bf Remark:} #1\\}
\def\example #1{\noindent{\bf Example:} #1\\}

\long\def\claim #1 #2{\bigskip\noindent{\bf Claim {#1}} {\it #2}\bigskip}
\def\xclaim #1 #2{\noindent{\bf Claim {#1}} {\it #2}\bigskip}
\newenvironment{proof}{\noindent{\bf Proof:}}{\hfill $\Box $\\}

\renewcommand{\thetheorem}{\arabic{section}.\arabic{theorem}}
\def\sqr#1#2{{\vcenter{\vbox{\hrule height .#2pt
              \hbox{\vrule width .#2pt height#1pt \kern#1pt
              \vrule width .#2pt} \hrule height .#2pt}}}}

\def\ncas #1 {\noindent {\bf Case #1.}\ }

\def\bipart #1 #2{\bigskip \noindent {\bf #1} {\it #2}}
\def\xbipart #1 #2{\noindent {\bf #1} {\it #2}}
\def\iipart #1 #2{\bigskip \noindent {\it #1} {\it #2}}
\def\xiipart #1 #2{\noindent {\it #1} {\it #2}}
\def\brpart #1 #2{\bigskip \noindent {\bf #1} {#2}}
\def\xbrpart #1 #2{\noindent {\bf #1} {#2}}
\def\irpart #1 #2{\bigskip \noindent {\it #1} {#2}}
\def\xirpart #1 #2{\noindent {\it #1} {#2}}

\def\o {\overline}

\def\case #1{\bigskip\noindent{{\bf Case} {\em #1}:}}
\def\subcase #1{\bigskip\noindent{{\bf Subcase} {\em #1}:}}
\def\numcase #1 #2{\bigskip\noindent{{\bf Case #1} {\em #2}:}}

\def\nclaim #1 {\noindent{\bf Claim #1: }}

\def\obs #1 {\bigskip\noindent{\bf Observation #1: }} 

\def\mathy #1{\ifmmode {#1}\else{$#1$}\fi}
\def\o{\overline}

\begin{document}

\ifcase 0  

\title{The Weighted Matching Approach \\to
Maximum Cardinality Matching} 
\author{%
Harold N.~Gabow%
\thanks{Department of Computer Science, University of Colorado at Boulder,
Boulder, Colorado 80309-0430, USA. 
E-mail: {\tt hal@cs.colorado.edu} 
}
}

\date{March 10, 2017}

\maketitle
\def\today{\ifcase\month\or
January\or February\or March\or April\or May\or June\or
July\or August\or September\or October\or November\or December\fi
\ \number\day, \number\year}
\def\date#1.#2.{\ifcase#1\or
January\or February\or March\or April\or May\or June\or
July\or August\or September\or October\or November\or December\fi
\ #2, \number\year}
\def\ydate#1.#2.#3.{\ifcase#1\or
January\or February\or March\or April\or May\or June\or
July\or August\or September\or October\or November\or December\fi
\ #2, 199#3}
\def\nydate#1.#2.{\ifcase#1\or
January\or February\or March\or April\or May\or June\or
July\or August\or September\or October\or November\or December\fi
\ #2}
\def\doublespace{\multiply\baselineskip by3\divide\baselineskip by2%
                 \def\doublespace{}}
\def\bigdoublespace{\multiply\baselineskip by2%
                 \def\bigdoublespace{}}
\def\imp{\ifmmode {\ \Longrightarrow \ }\else{$\ \Longrightarrow \ $}\fi}
\def\rimp{\ifmmode {\ \Longleftarrow \ }\else{$\ \Longleftarrow \ $}\fi}
\def\ximp{\ifmmode {\Longrightarrow\ }\else{$\Longrightarrow\ $}\fi}
\def\xrimp{\ifmmode {\Longleftarrow\ }\else{$\Longleftarrow\ $}\fi}
\def\iff{\ifmmode {\ \Longleftrightarrow \ }\else{$\ \Longleftrightarrow \ $}\fi}
\def\xiff{\ifmmode {\Longleftrightarrow\ }\else{$\Longleftrightarrow\ $}\fi}
\def\tru{\ {\bf true}\ }
\def\fal{\ {\bf false}\ }
\def\wrt{\ {\it wrt}\ }
\def\endskip{\medskip}
\def\qed{$\Box$}
\def\qedn{\ \vrule width4pt depth-1pt height7pt }
\def\rqed{\hfill\hbox to 24 pt{\vrule width4pt depth-1pt
height7pt\hfil}\bigskip}
\def\rqedn{\hfill\hbox to 24 pt{\vrule width4pt depth-1pt height7pt\hfil}}
\def\log{\ifmmode \,{ \rm log}\,\else{\it log }\fi}
\def\con {\subseteq}
\def\pcon{\subset}
\def\firstnumstp#1 {\bigskip \noindent{\it Step} #1.\newquad}
\def\numstp#1 {\endskip\noindent{\it Step} #1.\newquad}
\def\newquad{\hskip1ex}
\def\stp#1.{\endskip
\noindent{\it #1 Step.}\newquad}
\def\firststp#1.{\bigskip
\penalty-1000
\noindent{\it #1 Step.}\newquad}
\def\cas#1 {\smallskip\noindent{\bf Case} #1.\ } 
%
%
%
\long\def\sec#1{\bigskip
\penalty-2000%
\noindent{\twelvebf #1}\par\ignorespaces\noindent\ignorespaces}
\def\aorbsec#1{\noindent{\twelvebf #1}}
\def\nsec#1{\penalty-2000%
\noindent{\bf #1\hfill\break}
\hbox to \parindent{\hfill}\ignorespaces}
\long\def\res #1. #2{\bigskip
\penalty-1000
\noindent {\bf #1.}\newquad%
#2 \bigskip}
\long\def\nres #1. #2{\bigskip
\noindent {\bf #1.}\newquad%
#2}
\def\pf{\noindent {\bf Proof.}\newquad}
\def\cont{\ifmmode\star\else$\star$\fi}
\def\+{\tabalign} 
\def\nskp{\def\bigskip{}}
\def\i{($i$) } \def\xi{($i$)}
\def\ii{($ii$) } \def\xii{($ii$)}
\def\iii{($iii$) } \def\xiii{($iii$)}
\def\iv{($iv$) } \def\xiv{($iv$)}
\def\pa{({\it a}) } \def\xpa{({\it a})} 
\def\pb{({\it b}) } \def\xpb{({\it b})}
\def\pc{({\it c}) } \def\xpc{({\it c})}
\def\hi{\hskip20pt\i} \def\hii{\hskip20pt\ii} \def\hiii{\hskip20pt\iii}
\def\ha{\hskip20pt\pa} \def\hb{\hskip20pt\pb} \def\hc{\hskip20pt\pc}
\def\tran{{\buildrel*\over\to}}
\def\n{\rlap{$\>/$}}
\def\({{\rm(}} \def\){{\rm)}}
\def\c#1{\lceil {#1} \rceil}
\def\f#1{\lfloor {#1} \rfloor}
\long\def\boxit#1{\vtop{\hrule
\hbox{\vrule\quad\vtop{\vskip5pt\hbox{#1}\vskip5pt}\quad\vrule}
\hrule}} 
\def\iboxit#1{\vtop{\hrule
\hbox{\vrule\quad\vtop{\vskip5pt\hbox{{\it #1}}\vskip5pt}\quad\vrule}
\hrule}} 
\def\x{\iffalse}
\def\b{\bigskip}
\def\set #1#2{\{ #1:#2 \}}
\def\pset #1#2{( #1:#2 )}
\def\h{\hskip20pt}
\def\hi{\advance\parindent by 20pt}

\def\o{\overline} 
\def\u{\underline}
\def\opn{\hangindent=40pt\hangafter=1}
\def\h{{\hskip 20pt}}
\def\v{\vfill}
\def\hi{\advance \parindent by 20pt}
\def\d{\cdot}
\def \il #1{\log^{(#1)} }
\def\al.{{\it add\_leaf}}
\def\alm{{\it add\_leaf}$\,$}
\def\O{o\hbox{-}smallest}
\def\os.{\ifmmode{ \o{\cal S} }\else{$\o {\cal S}$}\fi}
\def\oP.{\ifmmode{ \o{\cal P} }\else{$\o {\cal P}$}\fi}
\def\ot.{\mathy{ \o{\cal T} }}
\def\oG{\o G}
\def\oB{\o B}
\def\oE.{\mathy{\overline E}}
\def\p(#1,#2){\ifmmode p(#1,#2) \else{$p(#1,#2)$}\fi}
\def\op(#1,#2){\ifmmode \o{p}(#1,#2) \else{$\o{p}(#1,#2)$}\fi}
\def\lb{\ifmmode \,{ \rm log}_\beta \else{\it log XX }\fi}
\def\wh{\widehat}
\def\wx.{\ifmmode \wh x \else$\wh x$\fi}
\def\wy.{\ifmmode \wh y \else$\wh y$\fi}
\def\wz.{\ifmmode \wh z \else$\wh z$\fi}
\def\wv.{\ifmmode \wh v \else$\wh v$\fi}
\def\Px.{\ifmmode \wh x \else$\wh x$\fi}
\def\Py.{\ifmmode \wh y \else$\wh y$\fi}
\def\Pz.{\ifmmode \wh z \else$\wh z$\fi}
\def\Pv.{\ifmmode \wh v \else$\wh v$\fi}
\def\Pr.{\ifmmode \wh r \else$\wh r$\fi}
\def\Pr.{\ifmmode \wh r \else$\wh r$\fi}
\def\A.{\mathy{{\cal A}}}
\def\B.{\mathy{{\cal B}}}
\def\E.{\ifmmode {{\cal E}}\else{{$\cal E$}}\fi}
\def\F.{\mathy{\cal F}}
\def\H.{\mathy{\cal H}}
\def\M.{\mathy{\cal M}}
\def\P.{\mathy{\cal P}}
\def\S.{\ifmmode {{\cal S}}\else{{$\cal S$}}\fi}
\def\T.{\mathy{\cal T}}
\def\mathy #1{\ifmmode {#1}\else{$#1$}\fi}
\def\goin{\hspace{17pt}}

\begin{abstract}
Several papers have achieved time 
$O(\sqrt n m)$ for
cardinality matching,
starting from first principles. 
This results in a long  derivation.
We simplify the
task by employing well-known concepts for maximum
weight matching. We use  Edmonds' algorithm 
to derive the
structure of shortest augmenting paths.
We extend this to a complete algorithm for maximum cardinality matching
in time $O(\sqrt n m)$.
\end{abstract}
\fi 

\def\switch{0}
\ifcase \switch 
\iftrue
\def\NumTextLabelWidth{30pt}
\def\NumTextWidth{419.75pt}
\long\def\NumText (#1) #2 {\bigskip
\noindent
\hbox to \NumTextLabelWidth{(#1)\hfill}\parbox[c]{\NumTextWidth}{#2}
\bigskip}

\def\myif #1{{\bf if} {\em {#1}} {\bf then}}
\def\myelif #1{{\bf else if} {\em {#1}} {\bf then}}
\def\myifel #1 #2 #3 {{\bf if} {\em {#1}} {\bf then} {\em {#2}} {\bf else} {\em {#3}}}

\def\myfor #1{{\bf for} {\em {#1}} {\bf do}}
\def\mycom #1{{\tt /* #1 */}}

\def\sm.{\mathy{\cal S^-}}

\section{Introduction}
The most efficient known algorithms for cardinality matching
on nondense graphs 
achieve time $O(\sqrt n m)$.
The best known of these algorithms are not readily accessible:
Micali and Vazirani were first to present such an algorithm
\cite{MV} but proving it correct has met difficulties
\cite{V1,V2}. 
Gabow and Tarjan present a complete development but only at the end
of a long paper with a different goal,
a scaling algorithm for weighted matching \cite{GT89}.
Similarly
Goldberg and Karzanov develop a new framework for flow and matching problems
(``skew symmetric matchings'') and again the cardinality matching algorithm
requires mastery of this framework.
Each of these papers tackles a difficult subject
from first principles.

This paper presents an accessible
matching algorithm with time bound  $O(\sqrt n m)$. 
We 
simplify the task by taking advantage of the well-established
theory for maximum weight matching.
We include a  review of
Edmonds' weighted matching algorithm \cite{E} but
still it  is helpful
to be familiar with the algorithm.
Complete treatments are in various texts
e.g., \cite{CCPS, L, PS,S}.

At first glance maximum weight matching seems to offer little insight
to the problem. 
However the dual variables for weighted matching
reveal structure that can either be used directly or must be rederived
in a presentation from first principles.
More importantly we show that a judicious choice of edge weights maps
a large piece of the puzzle into simple properties of weighted matching.
 
The use of weighted matching for cardinality matching
was introduced by Gabow and Tarjan \cite{GT89}. 
Their cardinality matching algorithm
uses a relaxation of the linear program duals that 
is helpful for scaling. 
Our algorithm is based on the LP duals. In that respect it
differs from \cite{GT89} at the structural  level. 

However to make our presentation complete, at the data structure level
we  use the depth-first 
search procedure of \cite{GT89}, especially because of its simplicity.
We also give a more detailed analysis of the procedure's correctness 
than \cite{GT89},
in an appendix.
An alternative for this part of the algorithm is the
double depth-first search algorithm of Micali-Vazirani \cite{MV}.
We use Edmonds' algorithm to deduce the key structure
for this algorithm, in another appendix.

The paper presents
our algorithm in a top-down fashion.
Section \ref{ApproachSec} gives the overall approach, using two Phases. 
Section \ref{Phase1Sec}  uses 
Edmonds' algorithm to implement Phase 1.
Section \ref{Phase2Sec} restates the Gabow-Tarjan algorithm for Phase 2
\cite{GT89}, with a complete correctness proof in Appendix
\ref{FAPApp}. Section \ref{DataStructureSec} 
gives details of the data structure
that achieve our desired time bound. Appendix \ref{DSHApp}
uses Edmonds' algorithm to prove  existence of
the starting edge 
for the double depth-first search of the Micali-Vazirani algorithm
\cite{MV}.

\paragraph*{Terminology}
The given graph is always denoted as $G=(V,E)$. 
An edge in a minor of $G$ is denoted as its preimage, i.e., 
$xy$ for $x,y\in V$.
For a set of vertices $S\subseteq V$, 
$\gamma(S)$ denotes the set of edges with both ends
in $S$.

A {\it matching} $M$ on a graph is a set of vertex-disjoint edges.
$M$ has {\em maximum cardinality} 
if it has the greatest possible number of edges.
Let each edge $e$ have a real-valued
{\it weight} $w(e)$. 
For a set of edges $S$ define $w(S)=\sum_{e\in S} w(e)$.
$M$ has {\em maximum weight} if $w(M)$ is maximum.

For an edge $xy\in M$ we say $x$ and $y$ are {\em mates}.
A vertex is {\it free} if it is not on any matched edge.
An {\it alternating  path} is a vertex-simple path 
whose edges are alternately matched and unmatched.
(Paths of 0 or 1 edges are considered alternating.)
An {\it augmenting path  P} is an alternating path joining two distinct free
vertices.
To {\it augment the matching along $P$} means to enlarge the matching $M$ 
to $ M \oplus P$ (the symmetric difference of $M$ and $P$).
This gives a matching with one more edge.

\section{The approach}
\label{ApproachSec}
A {\em shortest augmenting path}, or {\em sap}, is an augmenting path
of shortest length possible.
Hopcroft and Karp \cite{HK} and independently Karzanov \cite{K}
presented an efficient approach to finding a maximum cardinality matching:
Repeatedly find a maximal collection of
vertex-disjoint {\em saps} and augment the matching with them.
This algorithm repeats only $O(\sqrt n)$ times \cite{HK,K}. 
Fig.\ref{HiAlg} gives our implementation of this approach.

\begin{algorithm}
\DontPrintSemicolon

$M \gets\emptyset$

{\bf loop}

\Indp

\tcc{Phase 1}

\myfor {every edge $e$} 
$w(e)\gets $ \myifel {$e\in M$} {$2$} {$0$}
 
execute a search of Edmonds' weighted matching algorithm

\myif {no augmenting path is found}  {\bf halt}
\mycom {$M$ has maximum cardinality}

form the graph $H$ of permissible edges

\bigskip

\tcc{Phase 2}
 
$\P. \gets $ a maximal set of vertex-disjoint augmenting paths in $H$

augment $M$ by the paths of \P.

\caption{The high-level cardinality matching algorithm.}
\label{HiAlg}
\end{algorithm}

We will show that Edmonds' algorithm halts having found an {\em sap}
(unless there is no augmenting path).  Furthermore it provides the
information needed to construct the graph $H$, which has the property
that the {\em saps} of $G$ correspond 1-to-1 to the augmenting paths
in $H$.  We will complete the algorithm using one of several known
algorithms to find the maximal set of augmenting paths \P..
Every iteration will use $O(m)$ time so the entire algorithm uses
$O(\sqrt n m)$ time.

\section{Phase 1 via  Edmonds' algorithm}
\label{Phase1Sec} 
Section \ref{EdAlgSec} states the simplified version of Edmonds' algorithm that we use.
It also
 gives a brief review of blossoms.
Section \ref{EdPropSec}
characterizes the augmenting paths that are found.
Section \ref{Phase1DetailSec} gives the algorithm for
Phase 1 and proves the graph $H$ has the key property.

\subsection{Edmonds' weighted matching algorithm}
\label{EdAlgSec} 
Fig.\ref{EdAlg} gives pseudocode for a search of Edmonds' weighted
matching algorithm \cite{E}. The figure
omits code 
that is never executed in the special case of Edmonds'
algorithm that we use.
Omissions are indicated by
comments. (Basically they result from the fact that 
there are no blossoms at the start of a search.)
We will explain the code of the figure and then give a precise statement
of the assumptions that simplify the code (see (A) below).
In the figure $M$ denotes the matching, $B_x$ is the blossom containing vertex $x$, see below.

The algorithm is illustrated in 
Fig.\ref{Phase1StepsFig}. 
Free vertices are square and matched edges are heavy.
The dashed edges of Fig.\ref{Phase1StepsFig}(f) form an augmenting path.
The numbers in Fig.\ref{Phase1StepsFig} are defined below.

\begin{algorithm}[t]
\DontPrintSemicolon

make every free vertex the (outer) root of an \os.-tree\;

\tcc{general algorithm also makes free blossoms into roots}

{\bf loop}

\Indp

\myif {$\exists$ tight edge $e=xy$ with $x$ outer, $B_x\ne B_y$}

\goin\myif {$y\notin V(\S.)$} {\tt /* grow step */}

{\goin\goin add $xy, yy'$ to \S., where $yy'\in M$\;

\goin\goin \tcc{general algorithm adds blossom $B_y$ \& its matched blossom}
}

\goin\myelif{$y$ is outer}

\goin\goin\myif{$x$ and $y$ are  in the same search tree} 
\mycom {blossom step}

\goin\goin\goin merge all blossoms in the fundamental cycle of $e$ in \os.\;

\goin\goin{\bf else} \mycom {$xy$ plus the \os.-paths to $x$ and  $y$ form
an augmenting path}

\goin\goin\goin {\bf return} \mycom {general algorithm proceeds to
augment the matching}


\mycom {general algorithm may expand an inner blossom}

\lElse {\tcc{dual adjustment step}
{\Indp{

$\delta=\min
 \set{y(e)-w(e)}{e=uv \mbox{ with $u$ outer, } v\notin V(\S.)} 
\ \cup $

\goin\goin\hskip4pt $\set{(y(e)-w(e))/2}{e=uv \mbox{ with $u,v$ in distinct
outer blossoms}}$

\tcc{general algorithm includes set for  duals of inner blossoms}

\myif{$\delta=\infty$} 
{{\bf return}  \tcc{$M$ has  maximum cardinality}}
\lFor{every vertex $v\in V(\S.)$}
\goin\goin\lIf{$v$ is inner}{$y(v)\gets y(v)+\delta$
{\bf else} $y(v)\gets y(v)-\delta$}

{\tt
/* general algorithm changes duals of blossoms.\ in particular\\

\goin $z(B)\gets z(B) +2\delta$ 
for every maximal outer blossom $B$.\ */

}

}}}


\caption{Simplified search of Edmonds' algorithm.}
\label{EdAlg}
\end{algorithm}

\begin{figure}[th]
\centering
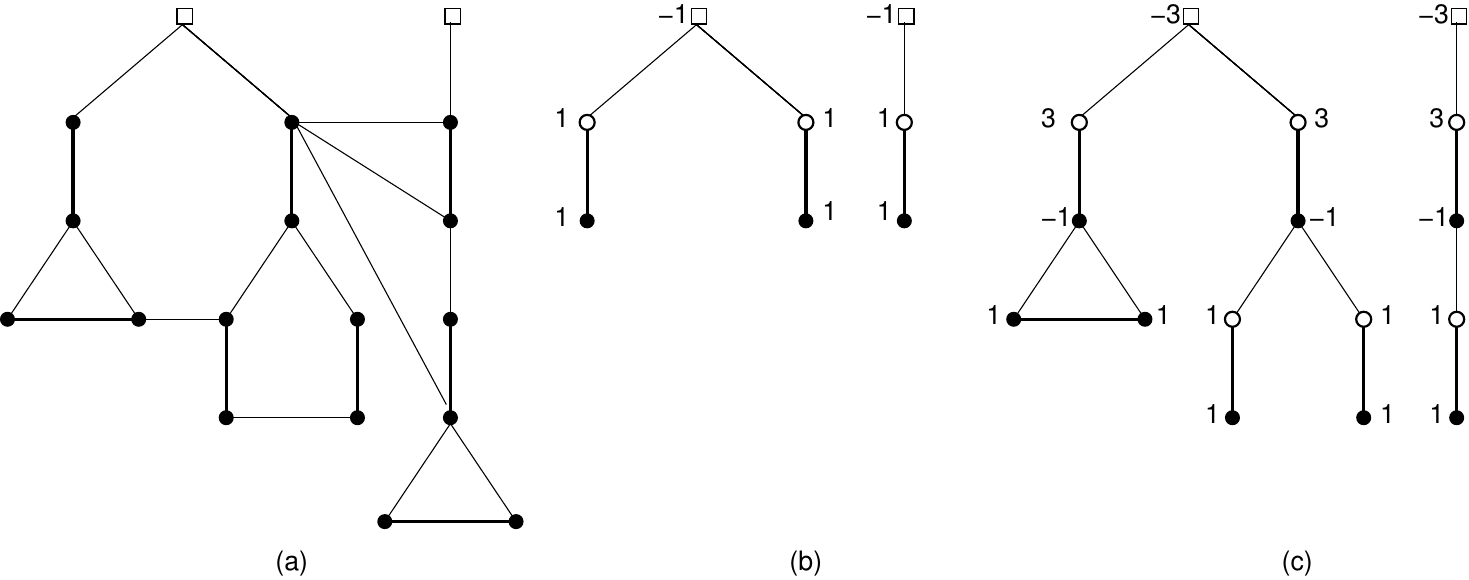

\vskip 20pt

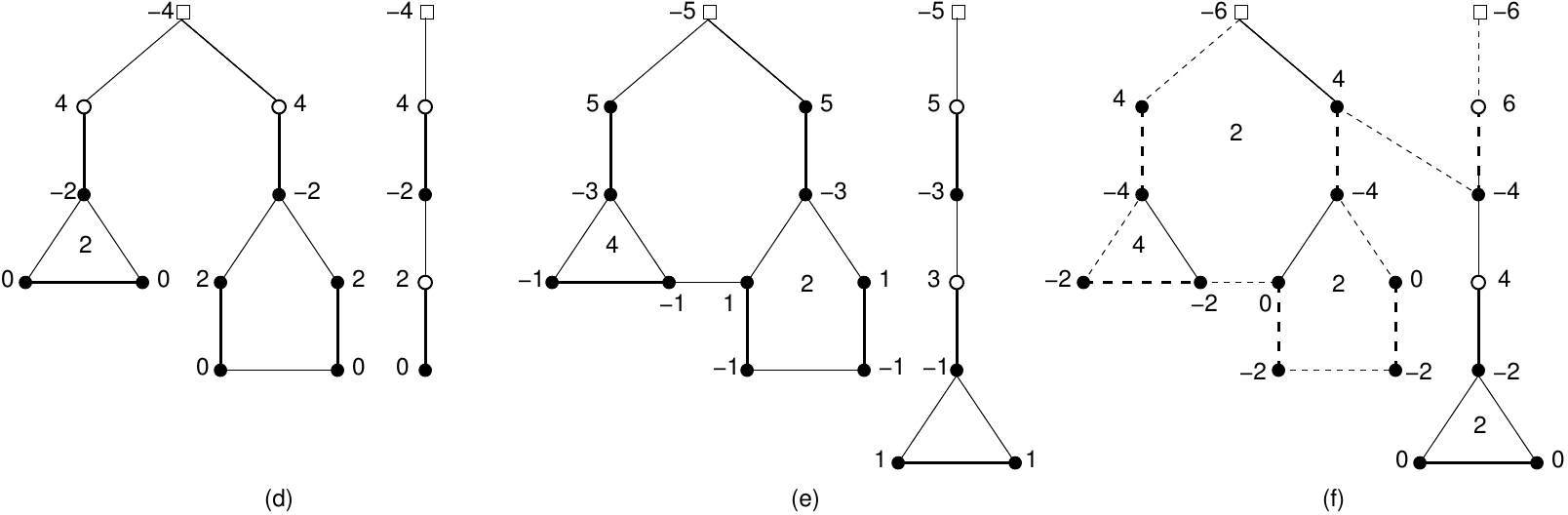
 \caption{(a) Example graph for finding an {\em sap}. 
Unmatched edges weigh  0, matched edges weigh 2.
(b)--(f) Search of Edmonds' algorithm.
(b)--(e) show \S., with  inner vertices drawn hollow,
and the $y$ and $z$ dual variables. (f) shows the  augmenting path as dashed edges.}
 \label{Phase1StepsFig}
 \end{figure}

The algorithm builds a ``search structure''
$\cal S$, a subgraph of $G$. 
(Fig.\ref{Phase1StepsFig}(b)--(e) show the various \S. structures.)
It also
maintains \os.,
the subgraph \S. with various sets (called ``blossoms'', defined below)
contracted. \os. is a  forest. Its roots are
the free vertices and contracted sets that each contain a free vertex.
\os. is an ``alternating forest'' \cite{E}: Any path from a node to
the root of its tree is alternating.
A node of \os.
is {\em inner}
(respectively {\em outer}) if its path starts with an unmatched (matched) edge.
An inner node is always a vertex of $G$. An outer node can be a contracted
set. Any vertex of $V$ in such a set is also called {\em outer}.

The algorithm adds new vertices and edges to \S. in a
{\em grow step}, which adds an edge from an outer vertex to a new inner
vertex $y$. It also adds the matched edge from $y$ to its mate,
a new outer vertex. 
(Fig.\ref{Phase1StepsFig}(b) shows the result of
three grow steps.) 

Suppose the algorithm discovers an edge
$e$ joining two outer vertices. If $e$ joins two distinct trees of \os.
it completes an augmenting path. An {\em augment step}
enlarges the matching. 
(Fig.\ref{Phase1StepsFig}(f) shows the algorithm's 
augmenting path as dashed edges.)

If $e$ joins nodes in the same tree of \os.
a {\it blossom step} is done:
$e$ is  added to \S., and
the fundamental cycle $C$ of $e$ in forest \os. is contracted.
The vertices of $V$ belonging to contracted nodes on $C$
form a {\em blossom}.
(A blossom step is  executed in 
each of Fig.\ref{Phase1StepsFig}(c), (d), and (e).)

To describe the last step, recall
that the algorithm is based on Edmonds' formulation of weighted matching as
a linear program \cite{E}. Each vertex $v\in V$ has a dual value
$y(v)$. Each blossom $B\con V$ 
has a nonnegative dual value
$z(B)$. The duals
{\em dominate} an  edge $uv\in E$ if
\[
y(u) +y(v)  + \sum_{v,w\in B} z(B) \ge w(e).
\]
\noindent
$uv$ is 
{\em tight} if equality holds in this constraint.
The algorithm maintains duals so that when 
$z$ values are included, every edge is always dominated.
Furthermore every matched edge is tight,
and an edge triggering a grow, blossom or augment step must be tight.
We use the common convention that for
$e=uv\in E$,  $y(e)$ denotes
$y(u)+y(v)$. Note that the test for tightness in Fig.\ref{EdAlg}
has $B_x\ne B_y$ and so  does not require $z$ values.

When no grow, blossom, or augment step can be done,
the algorithm makes progress by executing a
{\it dual adjustment step}.
It modifies  dual variables
so one or more of the other steps can be performed.
If this cannot be done the current matching has maximum cardinality.
The modification maintains the invariant that every edge of \S. is tight.

In Fig.\ref{Phase1StepsFig} each matched (respectively unmatched) edge
weighs 2 (0).
Each vertex is labelled with its $y$-value.
Each blossom is labelled with its $z$-value, 
even though these values are not
recorded in our algorithm. 
The label is in the interior of the blossom, and only included when $z$ is nonzero. For example the dual adjustment 
at the end of Fig.\ref{Phase1StepsFig}(c)
increases the 
$z$-value of the triangular blossom from 0 to 2. This dual increases
to 4 after part (d), and does not change in the dual adjustment
after part (e).

We conclude this section by stating the assumptions
that allow our simplifications to the general weighted matching algorithm:

\NumText (A) {%
$\bullet$ The algorithm begins  with a matching $M$ and  no blossoms.

$\bullet$ The algorithm begins 
with dual variables $y(v), v\in V$ that 
dominate  every edge and are tight on every matched edge (there are no $z$ variables).

$\bullet$
The algorithm does not need to 
track $z$ values of blossoms (since there will be no subsequent search).

}

\paragraph*{Blossoms}
It is convenient to consider every vertex as a blossom.
But these singleton vertices do not have a $z$ dual.
The notation $B_x$ for $x\in V$ denotes the maximal blossom containing $x$.

Any blossom $B$ has a {\em base vertex}: The base vertex of a singleton
$B=\{v\}$ is $v$. The base vertex of a blossom constructed
in the blossom step of Fig.\ref{EdAlg}
is defined as follows:
The fundamental cycle of $e=xy$ in \os.
contains a unique node
 of minimum depth -- the nearest common ancestor
$a$ of $x$ and $y$. The base vertex of the new blossom is the
base vertex of $a$. (In Fig.\ref{Phase1StepsFig}(b)--(f) the base of any
blossom is  the vertex 
 closest to the root. For instance in part (e)
the left subgraph is a blossom whose base vertex is the root.)

Note that
the base vertex of an arbitrary blossom
$B$ is the unique vertex of $B$ that is not
matched to another vertex of $B$.  It may be
free or matched to a vertex
not in $B$.

Any
 blossom $B$
has  a natural representation as an ordered tree $R_B$.
The root is a node corresponding to $B$.
The leaves
correspond to the vertices of $V$ in $B$.
Any interior node  corresponds to a blossom $B'$ formed in the blossom step
of Fig.\ref{EdAlg}. Let 
the fundamental cycle $C$ of that step consist of blossoms
$C_i$, $i=0,\ldots, k$. Here $C_0$ contains the base vertex of $B'$,
and the indexing corresponds to the order of the blossoms in 
 a traversal of $C$ (in either direction). The children of $B'$ 
correspond to $C_i$, $i=0,\ldots, k$, in that order.
In addition $R_B$ records the edge $c_ic_{i+1}$ that joins
each child $C_i$ to the next child  $C_{i+1}$ (taking $k+1$ to be 0).
These edges are alternately unmatched and matched. 
Each matched edge has one of its ends a vertex of $V$ (i.e., $C_i=c_i$).
Each end is  the base vertex of its blossom.
The base vertex of
$B$ 
 is also recorded in $R_B$, since it is not determined
by a matched edge of $C$.

Note that the edges $c_ic_{i+1}$ of $R_B$ all belong to
$E(\S.)$ and are tight. Also
$R_B$ has $O(|B|)$ nodes.
This follows since each $B'$ has $k+1$ children and $k/2$ matched edges.
So the $k+1$ children can be associated with $2\cdot k/2=k$ vertices
of $V$ in $B$ ($k\ge 2$).

Any vertex $v\in B$  has an even-length alternating path
$P(v,b)\con E(\S.)$ that starts with the matched edge at $v$ and ends at the base
$b$.
(The exception is $P(b,b)$, which has no edges.) $P(v,b)$ is
specified recursively using $R_B$, as follows. 
Let $C_i, i=0,\ldots,k$ be the children of the root $B$. Let
$v$ belong to child $C_j$.
The path $P(v,b)$ passes through $C_h$ for $h=j,j+1,\ldots,k,0$
if $c_jc_{j+1}$ is matched,
else $h=j,j-1,\ldots, 1,0$ if
$c_{j-1}c_{j}$ is matched.
Applying this description recursively to the children $C_h$
gives the entire path $P(v,b)$. 
(Note that 
$P(v,b)$ traverses
some recursive subpaths $P(v',b')$ in reverse order,
from $b'$ to $v'$. But  the algorithm does not require this order.)

As an example, the augmenting path of
Fig.\ref{Phase1StepsFig}(f)
 traverses the blossom of Fig.\ref{Phase1StepsFig}(e)
on 
the path
$P(v,f)$ of length 10.

The paths $P(v,b)$ are used (in the complete version of Edmonds'
algorithm, as well as our algorithm in Phase 2) to augment the matching.
Note that the order of edges in $P(v,b)$ is irrelevant for this operation,
since we are simply changing matched edges to unmatched and vice versa.
Also note that every edge of $P(v,b)$ is tight, since $P(v,b) \con E(\S.)$.
So augmenting keeps every matched edge  tight.

\subsection{Properties of Edmonds' algorithm}
\label{EdPropSec}
Implicit in Edmonds' algorithm is that it finds a maximum weight
augmenting path. This section proves this and characterizes 
the structure of all maximum weight augmenting paths.
We make this assumption on the initialization:

\NumText (A$'$) {
The initial $y$ function is constant, with no unmatched edge tight.}

\noindent
A constant $y$ is the usual initialization of Edmonds' algorithm.
Having no unmatched tight means there is at least one dual adjustment --
this assumption simplifies the notation.

As usual
define the weight of a path $P$ to be 
\[w(P)=w(P- M)-w(P\cap M).\]

\begin{lemma}
\label{EdmondsAPLemma}
At any point in Edmonds' algorithm, 
any augmenting path $P$ and any free vertex $f$ have
$w(P)\le 2y(f)$.
\end{lemma}

\begin{proof}
An edge $rs$ is dominated if it is unmatched, i.e.,
\begin{equation}
\label{WLeHEqn}
w(rs) \le y(r)+y(s)+\sum_{r,s\in V(B)} z(B),
\end{equation}
and equality holds if $rs$ is matched. 
So replacing every term $w(rs)$ in the definition of
 $w(P)$ by the right-hand side
of \eqref{WLeHEqn} gives an upper bound on $w(P)$.
Any interior vertex  of $P$, say $r$, 
has $y(r)$ appearing in both
$w(P-M)$ and  $w(P\cap M)$.
Hence the $y$ terms in the upper bound sum to $2y(f)$.
So it suffices to show the $z$ terms have nonpositive sum.

Let $B$ be any blossom  and let $b$ be its base.
Every vertex of $B$ except $b$ has its mate contained in $B$.
$P$ is not contained in $\gamma(B)$ since $P$ contains two free vertices.
Consider a maximal length subpath $S$ of $P\cap \gamma(B)$.
$S$ is alternating.
There are two cases:

\case {$S$ has $b$ at one end}
Following edges starting from $b$ shows
$S$ ends at a matched edge of $\gamma(B)$.
So $S$ has even length, and its edges make no net contribution
of $z(B)$ terms to the upper bound. 

\case {Neither end of $S$ is $b$}
The first and last edges of $S$ are matched.
Thus $S$
makes a net contribution of 
$-z(B)\le 0$ to the upper bound. 

We conclude the total contribution of $z$ terms to the upper bound
is $\le 0$, as desired.
\end{proof}

Call a blossom $B$ {\em positive} if $z(B)>0$.
$B$ is positive iff it was formed before the last dual adjustment.

The following corollary refers to the end of Edmonds' search --
tightness refers to the final duals, and blossoms are as defined 
over the entire algorithm.

\begin{corollary}
\label{EdmondsAPCor}
An augmenting path $P$ has  maximum weight iff all its edges are tight
and for every positive blossom $B$, $P\cap \gamma(B)$  is an even-length 
alternating path.
\end{corollary}

\begin{proof}
The if direction follows
 from the proof of the lemma. Specifically
 the proof  implies that $P$ achieves the upper bound of the lemma,
i.e., $w(P)=2y(f)$, if
every edge of $P$ is tight and $P$ traverses positive blossoms as specified
in the lemma. 

In particular the augmenting path $A$ found in
 the search algorithm satisfies these sufficient conditions.
(In fact $A$ traverses {\em every}  blossom
as in the lemma.) Thus $A$ is a maximum weight augmenting path.

This implies  an augmenting path has maximum weight iff its weight is $2y(f)$. 
So the proof of the lemma implies the conditions of the corollary must
hold for the final dual variables and the positive blossoms of the algorithm.
\end{proof}

\subsection{Phase 1}
\label{Phase1DetailSec}
As in Fig.\ref{HiAlg} an edge weighs 2 if it is matched, else 0.
We initialize Edmonds' algorithm by setting every $y$ value to 1.
This makes every matched edge tight, 
every unmatched edge dominated but not tight, 
and $y$ constant.
Furthermore
in every iteration of Fig.\ref{HiAlg} Edmonds' algorithm starts afresh with no
blossoms. So the assumptions of (A) and (A$'$) 
hold.

Assuming the matching is not maximum cardinality,
Edmonds' algorithm halts with a maximum weight augmenting path.
Any augmenting path $P$ of length $|P|$ has 
\begin{equation}
\label{LWEqn}
|P|=-w(P)+1.
\end{equation}
 So an augmenting path has maximum weight iff it is an
{\em sap}.

The last step of Phase 1 constructs $H$, the graph whose augmenting paths
correspond to the 
{\em saps} of $G$.
$H$ is formed from $G$ by contracting every positive blossom,
and keeping only the tight edges that join distinct vertices.%
\footnote{Edmonds' search may end before exploring some
tight edges. This does not affect the definition of $H$.}
As usual each edge of $H$ records its preimage in $E$,
allowing an augmenting path in $H$ to be converted to its preimage in $G$.

Phase 1 requires that
the augmenting paths in $H$ are precisely
the images of the {\em saps} of $G$.
We will show this using Corollary \ref{EdmondsAPCor}.

\case {$P$ is an {\em sap} in $G$} Consider any positive blossom $B$
with $P\cap \gamma(B)\ne \emptyset$. The corollary implies
$P\cap \gamma(B)$ is a path that starts with a matched edge
and ends with an unmatched edge incident to the base vertex $b$ of $B$.
So $P$ either
contains exactly 1 edge incident to $B$ (if $b$ is free),
or 2 edges incident to $B$, one being the matched edge incident to $b$.
In both cases the image of $P$ in $H$ is an alternating path. This makes
$P$ an augmenting path in $H$.

\case {$P$ is augmenting in $H$} Consider any blossom node
$B$ on $P$. $P$ contains an unmatched edge
incident to some vertex $v\in B$. Letting $b$ be the base vertex of $B$,
either $b$ is free or $P$ contains the matched edge incident to $b$.
In both cases we can add the $P(v,b)$ path through $B$. Doing this
for every $B$ on $P$ yields a path in $G$, with all edges tight.
The corollary shows $P$ is an {\em sap}.

\section{Phase 2}
\label{Phase2Sec}
 We can find a maximal set of augmenting paths in $H$
using the double depth-first search of Micali and Vazirani \cite{MV},
or the algorithm of Goldberg and Karzanov \cite{GK}, or the depth-first search
of Gabow and Tarjan \cite{GT89}. To make this paper complete 
Fig.\ref{ScaleAlg}
restates
the algorithm of \cite{GT89}.
The {\em find\_ap} algorithm is illustrated in Fig.\ref{SAPFig}. 
This section discusses the idea of the algorithm
and gives the high-level analysis,  for both correctness and the
linear time bound $O(m)$.
Appendix \ref{FAPApp}   completes the analysis.
The development 
is similar to \cite{GT89} but includes more details.

Note that the search of Micali-Vazirani \cite{MV}
requires an additional property of $H$, which is proved in
 Appendix \ref{DSHApp}.

\begin{algorithm}[t]
\DontPrintSemicolon

{\bf procedure} {\em find\_ap\_set}

initialize \S. to an empty graph and \P. to an empty set

{\bf for} {\em each vertex $v\in V$} {\bf do} $b(v)\gets v$
\tcc{$b(v)$ maintains the base vertex of $B_v$}

{\bf for} {\em each  free vertex $f$} {\bf do}

{\Indp 

{\bf if} {$f\notin V(\P.)$} {\bf then}

{\Indp

add $f$ to \S. as the root of a new search tree

{\em find\_ap}$(f)$

}}

\bigskip

{\bf procedure} {\em find\_ap}($x$)
\tcc{$x$ is an outer vertex}

\lnl{FLine}{\bf for} {\em each edge $xy\notin M$} {\bf do} \tcc{scan $xy$ from $x$}

\Indp

{\myif {$y \notin V(\S.)$}

{\goin \myif {$y$ is free} \tcc{$y$ completes an augmenting path} 

\lnl{ALine}{\goin\goin add $xy$ to \S. and add path $y P(x)$ to \P.\;

\goin\goin terminate every currently executing recursive call to {\em find\_ap}
}

\goin {\bf else} \tcc{grow step}

\goin\goin add $xy, yy'$ to \S., where $yy'\in M$

\goin\goin {\em find\_ap$(y')$}
}

}

\lnl{TLine}{\bf else} \myif {$b(y)$ is an outer proper descendant of $b(x)$ 
in \sm.} \tcc{blossom step}

\goin \tcc{equivalent test:\ $b(y)$ became outer strictly after $b(x)$}

\goin let $u_i$, $i=1,\ldots, k$ be the inner vertices in $P(y,b(x))$, ordered so $u_i$ precedes $u_{i-1}$


\goin \myfor {$i\gets 1\ {\bf to}\ k$} 


\lnl{BLine}\goin\goin \myfor {every vertex $v$ with $b(v)\in \{u_i,u'_i\}$, where $u_i u'_i\in M$} $b(v)\gets b(x)$

\goin\goin \tcc{this executes the blossom step for $xy$.\
each $u_i$ is now outer$.$}

\goin \myfor {$i\gets 1\ {\bf to}\ k$}  {\em find\_ap$(u_i)$}

\lnl{RLine}{\bf return}

\caption{Path-preserving depth-first search.}
\label{ScaleAlg}
\end{algorithm}

\begin{figure}[th]
\centering
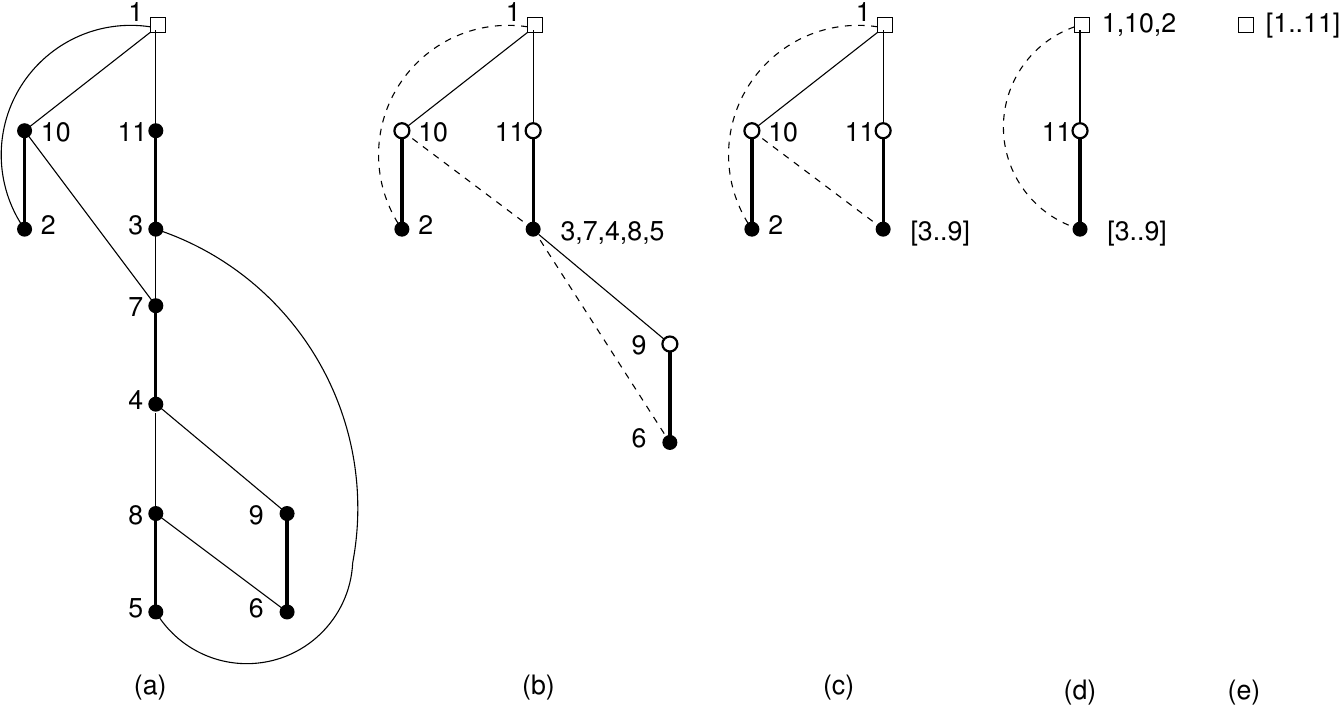
 \caption{(a) Example graph for {\em find\_ap}. 
Vertices are numbered in the order they become outer.
(b)--(e) \os. after each of the four blossom steps of {\em find\_ap}. Dashed edges are not in \os.. $[i..j]$ denotes the set of consecutive integers 
$\{i,\ldots, j\}$.}
 \label{SAPFig}
 \end{figure}

We first introduce a variant of previous notation 
that is used in the rest of the paper: \sm.
denotes  the subgraph of \S. consisting of
the edges added 
in grow steps. Clearly \sm. consists of trees that
span \S.. Also
\os. is a contraction of \sm..
We use \sm. to state 
ancestry relationships (e.g., line \ref{TLine} 
of Fig.\ref{ScaleAlg}). These relations essentially hold in
\os. but \sm. has the advantage of being more stable.
Note that \sm. is not alternating. 
(For instance in Fig.\ref{SAPFig}(a) add a
matched edge $aa'$ with an unmatched edge from 11 to $a$.)

Next we  review Fig. \ref{ScaleAlg}. 
{\em find\_ap} implements a search of
Edmonds'
cardinality matching algorithm (the algorithm of
Fig.\ref{EdAlg} with no dual variables, every edge is
tight). 
In line \ref{ALine}
$P(x)$ is the naturally defined alternating path in $H$ from $x$ to $f$.
Specifically $P(x)$ is formed from the \os.-path from $B_x$ to $B_f$,
by traversing every blossom of {\em find\_ap}
using the appropriate  $P(v,b)$ path.
In line \ref{BLine} $P(y,b(x))$ can be traversed by
simply using the \os.-path from $B_y$ to $B_{x}$.
Note this simple property of the base vertex function $b$:
$b(v)$ is always an \sm.-ancestor of $v$.

The idea for
{\em find\_ap} is to 
 search for an augmenting path depth-first,
making sure that
all vertices currently being explored remain on  the current search path.
This is an obvious property of ordinary depth-first search.
It is also desirable for finding disjoint augmenting paths:
When we find an augmenting path and delete it from further consideration,
no problems are created by partially explored vertices 
remaining in the graph -- these vertices all get deleted.

To achieve this property blossom steps must be 
scheduled carefully. 
Fig.\ref{ScaleAlg} does this by delaying blossom steps.
(In Fig.\ref{SAPFig} 
$mate(2)$ is the 10th vertex to become outer, not the third.)
Also the blossom step explores the new outer vertices $u_i$ in 
order of decreasing $P(u_i)$ length.
(Vertex 7 is explored before vertex 8.)

It is easy to see that {\em find\_ap} correctly implements the grow and
blossom steps of Edmonds' cardinality matching algorithm.  
 Specifically $\S., \os., 
b(x),$ and $ P(x)$ correspond to their definitions.  
However it is not immediately clear that 
{\em
find\_ap} does every possible blossom step.
This is proved in Appendix A
(Lemma \ref{CompletenessLemma}).

We close this section with a simple proof that
{\em find\_ap\_set} is correct, assuming 
Lemma \ref{CompletenessLemma}. Correctness means that
 {\em find\_ap\_set} halts with 
\P. a maximal set of vertex-disjoint augmenting paths.
The above remarks imply that each path of \P. is augmenting.
So the issue is to prove maximality -- no augmenting path of $H$
is vertex-disjoint from \P..

At any point in the execution of {\em find\_ap\_set}
call an outer vertex $x$ {\em completely scanned}
if {\em find\_ap}$(x)$ has returned in line \ref{RLine}.
The other possibility for an outer $x$ -- $x$ is not completely scanned
--
holds when either
{\em find\_ap}$(x)$ 
has not been called, or it is currently executing, or 
it was terminated because
an augmenting path was discovered.

The important invariant is 

\NumText (I)
{Whenever 
line \ref{FLine} of {\em find\_ap}$(x)$ scans an edge,
$P(x)$ contains every  outer vertex that has not been completely scanned.}

\noindent
For instance when 
vertex 8 scans edge 8,6,
vertex 7 is not in $P(8)$ but it has been completely scanned.
It is easy to see  invariant (I)
is preserved when a recursive call is made:
In the grow step $P(y')$ contains
$P(x)$; 
in the blossom step $P(u_i)$  contains $P(u_j), i<j \le k$, and also $P(x)$.
So (I) holds on entry to {\em find\_ap}$(x)$. Any vertex that becomes
outer after this is completely scanned
whenever control returns to 
{\em find\_ap}$(x)$.
So (I) holds throughout the execution of
{\em find\_ap}$(x)$, and throughout the entire algorithm.

(I) implies that any augmenting path
added to \P. contains every  outer vertex that has not been completely scanned.
Put the other way, when {\em find\_ap\_set} halts
every outer
vertex not on a path of \P. has been completely scanned.
We conclude the following properties when {\em find\_ap\_set} halts:

\bigskip

{\narrower

{\parindent=0pt

\i Any free vertex
$\notin V(\P.)$
is outer.

\ii Any edge $uv$ with $u$ outer and $u,v\notin V(\P.)$ 
has $v$ inner or 
$B_u= B_v$.

}

}

\bigskip
\noindent
Note that \ii depends on Lemma \ref{CompletenessLemma}.

Now consider any alternating path $P$ that is disjoint from
$V(\P.)$ and 
starts at a free vertex $f$. 
We claim that every vertex of $P$ is inner or outer. Furthermore
for any outer blossom $B$ 
with $B\cap V(P)\ne \emptyset$, the first vertex of $P$ in $B$ is the
base of $B$,
which is either $f$ or the mate of
an inner vertex.

To prove the claim we argue by induction on the length of $P$.
The claim holds if $P$ is just the starting vertex $f$, by \xi. 
Inductively assume
$P$ has reached an outer vertex $x$, having already reached $b(x)$.
If the next vertex $y$ is outer \ii shows $b(y)=b(x)$. 
If $y$ is not outer \ii shows it is inner. Furthermore $y$ is followed by
$mate(y)$, which is the base of its blossom $B_{mate(y)}$.

Now observe that
$P$ does not contain a 
free vertex $f'\ne f$.
In proof, $f'$ is outer. The claim shows $P$ must enter $B_{f'}$
on its base $b$ and
$b\ne f'$.
But $f'$ is the base of any blossom containing it.
Hence $f'\notin P$.
We have shown that no augmenting path
 is disjoint from
$V(\P.)$, as desired.

We close this section with two high-level properties of
{\em find\_ap\_set} that lead to its linear time bound.

First,
any edge is scanned at most twice, once from each end. 
In proof observe 
the call {\em find\_ap}$(x)$ 
occurs when $x$ becomes a new outer vertex.
Thus  {\em find\_ap} is called at most once
for any given vertex $x\in V$.

Second, the test of line \ref{TLine} is convenient to prove correctness
of the algorithm, but is not straightforward to implement. The test of
the comment is easily implemented. Appendix \ref{FAPApp}
shows the two tests are equivalent.

\section{The data structure}
\label{DataStructureSec}
This section gives the data structures that show 
Phases 1 and 2 both use time $O(m)$. Thus the entire algorithm runs in 
time $O(\sqrt n m)$. 

The matching is represented using an array {\em mate},
where any vertex $v$ on matched edge $vv'$ has $mate(v)=v'$.
The forest \sm. is represented using a pointer $\ell(v)$ for every
outer vertex $v\in V$.
Specifically a grow step adds $y$ and $y'$  to \sm. by setting
$\ell(y')=x$. 

We now discuss each step of our algorithm in turn, 
presenting the data structure for it and verifying that linear time is achieved.
We begin with Phase 1.
Grow steps are trivial.

\paragraph*{Blossom steps}
We find the fundamental cycle $C$ of edge $xy$
in  time $O(|C|)$
by climbing the paths from $B_x$ and $B_y$ to the root in parallel.
This is accomplished
 using
{\em mate} and $\ell$ pointers, and a data structure to find
the base of a blossom $B_x$ given an arbitrary vertex $x$.

The data structure is
the incremental-tree set-merging algorithm of 
Gabow and Tarjan \cite{GT85}.
The operation $find(x)$ returns the base vertex of $B_x$.
$union$ operations are used to  contract $C$.
The incremental-tree version of \cite{GT85} works on a tree that grows by addition of leaves, so it works correctly on \sm..
The time for $O(m)$ $finds$ and $O(n)$ $unions$ 
is $O(m+n)$ and the space is  $O(n)$.

\paragraph*{Dual adjustment steps}
It is well-known how to implement dual adjustment steps efficiently using
(a) the parameter
$\Delta$, defined as the current sum of all dual adjustment quantities $\delta$
so far; and (b) an appropriate priority queue.

Our choice for (b) depends on the fact that $\Delta$ is always an integer
$\le n/2$. We first prove the upper bound.
At any point in the algorithm
let $P$ be any augmenting path.
Lemma \ref{EdmondsAPLemma} shows
$w(P)\le 2y(f)$. Then using
\eqref{LWEqn} gives
\[
n-1 \ge |P|=-w(P)+1\ge -2y(f)+1.
\] Rearranging gives
$y(f)\ge 1 -n/2$. Since $y$ values are initially 1, this implies 
$\Delta\le n/2$.

To prove integrality 
we claim that
all the $y(v)$ values, $v\in V(\S.)$, are always
integers of the same parity.
(Since a free vertex $f$ has $y(f)=1-\Delta$ 
this implies $\Delta$ is integral.) The claim  holds 
initially by (A$'$).
In the definition of $\delta$ in Fig.\ref{EdAlg}, any edge $e$
has
$y(e)-w(e)=y(e)$  an even integer.
So the set defining $\delta$ consists of integers, and $\delta$ is 
integral. 
The adjustment of $y$ values
in Fig.\ref{EdAlg} keeps all values $y(v), v\in V(\S.)$ integers of
the same parity.

Now we sketch the algorithm and data structure.  The priority queue
consists of a collection of lists $L(d)$, each containing the edges
that become tight when $\Delta$ has increased to $d$.  To get the next
edge for the search lists $L(\Delta')$ are examined, for
$\Delta'=\Delta, \Delta+1, \ldots$, until an edge that triggers a
grow, blossom, or augment step is found.  This also gives the next
value of $\Delta$.

Suppose a grow or blossom step makes vertex $u\in V$ outer.  Every
unmatched edge $e=uv$ is scanned.  $e$ is added to list $L(\Delta+d)$
where $d$ is $y(e)$ if $v\notin V(\S.)$ or $y(e)/2$ if $v$ is outer.

\bigskip

The next two steps are implemented using the representation $R_B$ for
every maximal blossom $B$.  In the data structure for $R_B$, the
children of any node $B'$ form a doubly linked ring. Each link also
records the edge $xy\in E$ that joins the two subblossoms.

\paragraph*{Constructing graph {\boldmath $H$}}
The first task is to identify the blossoms that become vertices
of $H$, i.e., the blossoms that are maximal immediately before the last
dual adjustment. During Edmonds' search
each blossom is marked with
the value of $\Delta$ when it is formed.
When the search ends each representation $R_B$ is
traversed top-down, and blossoms with the final value of $\Delta$
are discarded.
The remaining roots of $R_B$ representations are the 
maximal positive
blossoms that become
vertices of  $V(H)$.

To construct $E(H)$, we continue the top-down traversal
and label
every vertex of $V$ with the $V(H)$ vertex
containing it. Then every edge $e\in E$ is scanned and added to $E(H)$
if it is tight and joins
distinct $V(H)$ vertices, at least one being outer.
($z$ values are not needed for these edges.)
$H$ is represented as an adjacency structure, and the vertex labels are used
to add $e$ to the two appropriate adjacency lists.
$e$ also records its preimage in $G$, so augments can be performed in $G$. 

\bigskip

Turning to Phase 2, the blossom base function $b$ is maintained as in
Phase 1 using incremental-tree set-merging.

\paragraph*{Computing {\boldmath $P(x)$} in {\em find\_ap}}
First observe that an $R_B$ data structure
allows a path $P(v,b)$ to be computed in time 
$O(|B|)$ by following the recursive procedure sketched in Section 
\ref{EdAlgSec}. 
This suffices for our applications, since an augmenting path 
passes through a maximal blossom at most once. 
(In Phase 2 
if an augmenting path passes through a blossom $B$, the  vertices of $B$
never get re-explored.)
 A more careful approach
computes $P(v,b)$ in time
linear in its length. The idea is not to start at the root
of $R_B$ (as in Section \ref{EdAlgSec})
but rather start at the node whose child contains
the matched edge incident to $v$.

{\em find\_ap} builds $R_B$ representations for the blossoms it creates.
They are  used to
compute the augmenting paths $yP(x)$.
These paths are then converted
to paths in $G$ by adding in $P(v,b)$ paths that traverse blossom-vertices of
$H$.  
This is done using the $R_B$ representations from Phase 1.

\bigskip

Combining our correctness proof 
and the above data structures giving linear time, our goal is achieved:

\begin{theorem}
A maximum cardinality matching can be found in time $O(\sqrt n m)$ and space
$O(m)$.
\end{theorem}


\fi
\or
\input intro

\input bnotes
\input nca
\input alpha
\input bmatch
\input code
\input strong
\fi

\ifcase 1 
\or

\fi


\clearpage

\setcounter{section}{0}
\renewcommand{\thesection}{\Alph{section}}
\renewcommand{\thetheorem}{\Alph{section}.\arabic{theorem}}

\setcounter{equation}{0}
\renewcommand{\theequation}{\Alph{section}.\arabic{equation}}

\section{Analysis of {\em find\_ap\_set}}
\label{FAPApp}
We will consider
\sm. to be an ordered tree, where grow steps add the children of
an outer vertex in left-to-right order. 
We also assume  \os. inherits this order.
(Thus 
in Fig.\ref{SAPFig}(b) vertex
3 became outer after vertex 2.)  

\begin{lemma}
\label{RightmostPathLemma}
Any outer vertex $r$
that is not completely scanned has $b(r)$ on the rightmost path of \sm..
\end{lemma}

\example {The blossom step of Fig.\ref{SAPFig}(b) 
makes vertex 8 outer but not completely scanned.
8 is not on the rightmost path
of \sm. but $b(r)=3$ is.} 

\begin{proof}
A simple induction shows that
when a vertex $x$ is made outer, 
every vertex $s\in P(x)$ has $b(s)$ on the rightmost path of \sm.. 
Recall invariant (I) 
says
whenever 
{\em find\_ap}$(x)$ scans an edge,
$P(x)$ contains every  outer vertex $r$ that has not been completely scanned.
So $b(r)$ is on the rightmost path of \sm..
\end{proof}

\begin{lemma}
\label{SRelationsLemma}
At any point in the algorithm, 
consider an edge $rs$ where $r$ is outer and $s\in V(\S.)$.
Either $s$ is inner and left of $r$ in \sm., 
or $b(r)$ and $b(s)$ are related in \sm..
\end{lemma}

\example {Consider edge 8,6 immediately after the blossom step forming
Fig.\ref{SAPFig}(b).
$b(8)=3$ is related to $b(6)=6$ in \sm..
But 8 itself is not related to 6 in \sm..
Neither is 5, the mate of 8.}

\begin{proof}
We will show that every grow and blossom step preserves the lemma.
We start with this preliminary observation: Once an outer vertex $r$ has been completely scanned,
any adjacent vertex $s$ is in \S..

Consider a grow step.
It adds new vertices $y,y'$, with $b(y)=y$, $b(y')=y'$.

\subcase {$r\ne y'$} If $r$ is not completely scanned,
Lemma \ref{RightmostPathLemma}
implies
$b(r)$ is related to both
$b(y)$ and $b(y')$. If $r$ is completely scanned the
preliminary observation shows $s$ was in \S. before the grow step.
So $s\ne y,y'$.
We conclude the lemma is preserved.

\subcase {$r=y'$} 
Since $r$ is the rightmost vertex of \sm.,
any vertex is either related to $r$ or to the left of $r$.
So the lemma holds if $s$ is inner.
If $s$ is outer it cannot be left of $r$,
since then the preliminary observation
shows $r$ was in \S. before the grow step.

\bigskip

Now consider a blossom step. We consider the possibilities for $r$.

\subcase {$r$ is  a vertex whose $b$-value is changed to $b(x)$}
(These are the vertices that enter $B_x$.)
Any vertex $s$ has $b(s)$ either related to $b(x)$ or left of $b(x)$
($b(x)$ is on the rightmost path by Lemma \ref{RightmostPathLemma}. 
So we can assume $b(s)$ is left
of $b(x)$. This implies $s$ is also left of $b(x)$.
$s$ cannot be outer (as before, $s$ is completely scanned,
so $r$ would be in \S. before it becomes a descendant of $b(x)$). 
So $s$ is inner as desired.

\subcase {$r$ is outer and $b(r)$ is not changed to $b(x)$} 
We can assume $s$ is a vertex that enters $B_x$. We can further assume
$b(r)$ is not related to $b(s)=b(x)$. 
Thus $b(r)$ is to the left of $b(x)$.
As before $r$ is completely scanned, 
making $s$ in \S. before it becomes a descendant of $b(x)$.
\end{proof}

\begin{lemma}
\label{InnerDescendantLemma}
At any point in the algorithm
let $t$ be an inner vertex, whose outer mate $t'$ is completely scanned,
and $s$  an inner \sm.-descendant of $t$.
A blossom step that makes $s$ outer
makes $b(s)=b(t)$.
\end{lemma}

\begin{proof}
Let $P$ be the \sm.-path from $t$ to $s$.
We prove the lemma by induction on $|P|$.

Among all the inner vertices on $P$,
let $u$ be the first to become outer in a blossom step.
(If there is more than one choice take $u$ as deep as possible.)
Let that blossom step be triggered by edge $xy$
where $b(x)$ is an ancestor of $b(y)$.
$u$ is an \sm.-ancestor of $b(y)$ and $b(x)$ is an ancestor of $u$.

Every outer vertex $r$ on $P$ is completely scanned, since $t'$ is.
(This follows since $r$ became outer while {\em find\_ap}$(t')$ 
was executing.)
So $b(x)$ is not on $P$.
Thus $b(x)$ is a proper ancestor of $t$. 
Letting $u'$ be the mate of $u$, the blossom step sets
\begin{equation}
\label{b1biEqn}
b(t)=b(u)=b(u').
\end{equation}

If $u=s$ we are done. Otherwise 
let $v$ be the inner vertex that follows $u$ on $P$.
Let $v'$ be its mate. As already mentioned, $v'$ is completely scanned.
So the inductive assertion holds for $v$ and $s$.
Consider the blossom step that makes $s$ outer.
Let $b_1$ denote the $b$ function at the end of this step.
The inductive assertion shows
\[b_1(v)=b_1(s).\]

 Since $v$ was inner, we also have $b_1(v)=b_1(u')$.
\eqref{b1biEqn} implies $b_1(t)=b_1(u')$. Combining equations gives
$b_1(t)=b_1(s)$. This completes the induction.
\end{proof}

\begin{lemma}
\label{CompletenessLemma}
At any point in the algorithm,
let $rs$ be an edge 
that has been scanned from both  its ends.
Then
$b(r)=b(s)$.
\end{lemma}

\begin{proof}
Whenever $r$ and $s$ are both outer,
$b(r)$ and $b(s)$ are related in \sm.
(Lemma \ref{SRelationsLemma}). 
Let $b(r)$ be an ancestor of $b(s)$
the first time both are outer.
Although $b(r)$ and $b(s)$ may change over time,
$b(r)$ will always be an ancestor of $b(s)$.

Consider the three possibilities for $s$  when
$rs$ is scanned from $r$. 

\case {$s$ is not in the search forest} A grow step makes $s$ an inner child of $r$. 
Eventually $s$ becomes outer in a blossom step. The new blossom has an outer base vertex, so the blossom includes $r$, i.e., $b(r)=b(s)$.

\case {$s$ is outer} Clearly a blossom step is executed, making $b(r)=b(s)$.
 
\case{$s$ is inner} When $rs$ is scanned from $r$ 
let
$t$ be first inner vertex on the \sm.-path from $b(r)$ to $s$.
When $r$ scans $rs$,
$t'=mate(t)$ is completely scanned.
Now apply Lemma \ref{InnerDescendantLemma} to $t$ and $s$.
The blossom step that makes $s$ outer makes 
$b_1(t)=b_1(s)$. Since $t$ has become outer $b_1(t)=b_1(r)$.
Thus $b_1(s)=b_1(r)$ as desired.
\end{proof}

To implement the algorithm efficiently we change  the test for a blossom step,
line \ref{TLine},
to the test  of the comment. We will show the two tests are equivalent, i.e.,

\bigskip

{$b(y)$ is an outer proper descendant of $b(x)$ 
in \sm. iff $b(y)$ became outer strictly after $b(x)$.}

\bigskip

To prove the if direction
assume $b(x)$ and $b(y)$ are both outer.
As blossom bases they both became outer when they were added to \S..
Edge $xy$ makes $b(x)$ and $b(y)$ related (Lemma \ref{SRelationsLemma}).
So if  $b(y)$ was made outer strictly after $b(x)$ it was added to \S.
after $b(x)$, i.e., it
descends from $b(x)$.
Thus the condition of line  \ref{TLine} 
holds.

The only if direction is obvious, since any vertex is added to \sm. after
its ancestors. 

\section{Searching from the middle}
\label{DSHApp}
The algorithm of Micali and Vazirani \cite{MV} is based on
a ``double depth-first search'': This search begins at
an edge $e=uv$.
It attempts to complete an augmenting path using disjoint paths from each of
$u$ and $v$
to a free vertex. 
This is done with
two coordinated depth-first searches, 
one starting at $u$, the other at $v$.

The key fact for this approach
is a characterization of the starting edge $e$.
We will begin by describing the conditions satisfied by $e$, using
our terminology.
Then we prove that any augmenting path
contains such an edge $e$. 
Then we discuss the implications of this structure, including
how the DDFS of \cite{MV} can be used for our Phase 2.

We start with terminology based on the state of the search
immediately before the last dual adjustment.
Let $T'$ be the set of edges of $G$ that are tight at that time.
Let $D_1 \cup D_2$ 
be the set of edges that become tight 
in the last dual adjustment, where
$D_1$ refers to a grow step and $D_2$ is for a blossom step.
So $e\in D_1$ has
$y'(e)=\delta$ with one end of $e$ outer and the other not in \S..
 $e\in D_2$ has
$y'(e)=2\delta$ with both  ends of $e$ outer. (Recall $w(e)=0$.)
Here $y'$ is the dual function right before the last dual adjustment,
and ``outer'' and \S. also refer to that time.

\def\mycenter #1 {\hbox to \hsize{\hfill{#1}\hfill}}

\begin{lemma}
Any maximum weight augmenting path can be written as\\
\mycenter{$P_1,Q,P_2$}
\newline
\noindent
where 

each $P_i$ is an even length alternating path from a free vertex
to an outer vertex, $P_i\con T'$,

$Q$ has the form 
$(e)$ with $e\in D_2$, or
$(g_1,e,g_2)$ with
$g_1,g_2\in D_1$.

\end{lemma}

\remark {Clearly $e$ is unmatched in the first form and  matched in the second.
Neither end of $e$ is in \S. in the matched form.}

\begin{proof}
Let $y$ be the final dual function.
The dual adjustment step shows that any free vertex $f$ has
$y(f)=y'(f)-\delta$.
As mentioned in the proof of Corollary \ref{EdmondsAPCor}
an augmenting path $P$ has maximum weight iff
$w(P)=2y(f)$. Thus
\begin{equation}
\label{OldYEqn}
w(P)=2y'(f)-2\delta. 
\end{equation}
Furthermore the corollary shows 
that for any positive blossom $B$,
$P\cap \gamma(B)$ is an even length alternating
path, so $z(B)$ makes no net contribution to $w(P)$.
Thus $P$ contains edges that are not tight in $y'$, in fact these edges belong
to $D_1\cup D_2$ and have
total slack
$2 \delta$.

Suppose $P$ contains an
edge $e\in D_2$.
Since 
$y'(e)=2\delta$, $P$ contains exactly 1 such edge.
The properties of the lemma for both $P_i$ and  $Q$ follow easily.

The other possibility is that $P$ contains exactly two edges 
$g_1,g_2\in D_1$. Each $g_i$ is unmatched and has an end $v_i\notin \S.$.
$P$ must contain a $v_1v_2$-subpath of edges in $T'$.
It must consist of just one edge $v_1v_2\in M$,
since unmatched edges with no end in \S. are not tight.
The properties of the lemma for both $P_i$ and  $Q$ follow.
\end{proof}

It may not be clear how
{\em find\_ap} succeeds in ignorance of this structure.
So we take  a more detailed look. 
We start with a simple 
fact:

\begin{proposition}
\label{InnerInnerProp}
No edge $uv\in T'$ joins 2 inner vertices.
\end{proposition}

\begin{proof}
A grow step that makes $u$ inner
has $y(u)=1$. Every subsequent
dual adjustment increases $y(u)$. So the search ends with
$y(u)\ge 1$. If $u$ and $v$ are inner then
$uv\notin M$, and $y(u)+y(v)\ge 2> w(uv)=0$.
\end{proof}

We now present a more detailed 
proof of the lemma.
Consider the search graph \os.  immediately before the 
last dual adjustment. \os. is a subgraph of $H$.
Define a path form similar to the lemma, as\\
\mycenter{$P,Q,P'$}
\newline
\noindent
where 

$P$ is an even length alternating path from a free vertex
to an outer vertex of \os., $P\con T'$;

$Q$ has the form of the lemma; 

$P'$ is an odd alternating path whose 
last edge is matched and last vertex is
inner in \os., $P'\con T'$.

\bigskip

Let $A$ be an alternating even-length path  in $H$
that starts at a free
vertex.
We claim that $A$ is a prefix of the above form.
Clearly the claim forces {\em find\_ap} to find
a path with the structure of the lemma.

We prove the claim inductively. Suppose an even length prefix $A'$ of
$A$
ends at vertex $u$, and the next two edges
of $A$ are $uv,vv'$
 with $uv\notin M \ni vv'$.

If $A'$ has length 0 then $u$ is free.
$A'$ has the form $P$. 

Suppose 
$A'$ has form $P$. 
There are three possibilities:

\subcase {$v$ is inner in \os.}
Its mate $v'$ is outer, so
form $P$ holds for the longer prefix.

\subcase {$v$ is outer in \os.} $uv$ joins two outer vertices
of \os.. Thus $uv\in D_2$.
The matched edge $vv'$ joins an outer vertex with an inner, so
$v'$ is inner. So the new prefix of $A$ has form $P,Q,P'$
for $P'=(vv')$.

\subcase {$v\notin \os.$} This makes $uv\in D_1$.
The new prefix has the form $P,g_1,e$ with $g_1,e$ as in $Q$.

\bigskip

Now suppose $A'$ has form $P,g_1,e$ with $g_1,e$ as in $Q$. 
Since $e\notin \os.$ and
$uv$ is tight, $uv\in D_2$. So $v$ is outer.
Thus $v'$ is inner. The new prefix has form $P,Q,P'$
($P'=(vv')$).

Finally suppose $A'$ has form $P,Q,P'$.
No end of an edge of $D_1\cup D_2$ is inner.
Since $u$ is inner this makes  $uv\in T'$.
Also $u$ inner makes
$v$ outer (Proposition \ref{InnerInnerProp}).
The new prefix ends with
edge $vv'\in M$ and $v'$ inner. Thus it has form $P,Q,P'$.
The induction is complete.

\bigskip

The lemma opens up the possibility of 
having a Phase 2 search start from an edge $e$ of type $Q$.
DDFS uses this strategy.

The Micali-Vazirani algorithm
uses DDFS in Phase 1 as well. This depends on the fact that
blossoms have a starting edge $e$ similar to the lemma.
(More precisely suppose a blossom step
creates a blossom $B$ with base $b$, with $v\in B$
a new outer vertex. Then $P(v,b)$ contains a unique subpath
of form $Q$ of the lemma.
This is easily proved as above,
e.g., use the second argument,
traversing the path $P(v,b)$ starting from $b$.)

Using DDFS in both Phases 1 and 2 makes the Micali-Vazirani algorithm
elegant and  avoids any
overhead in transitioning to Phase 2.

The proof of \cite{V2} that DDFS is correct is involved. Possibly
it could be simplified using the lemmas we have presented, as well as other
structural properties that weighted matching makes clear.
The following aspects of the finer structure of $H$ 
are not needed for our development but are used in \cite{V2}.

\cite{V2} defines  {\em evenlevel}$(x)$ as the length of a shortest
even alternating
path from a free vertex to $x$.
The proof of Lemma \ref{EdmondsAPLemma} shows 
any even alternating $fx$-path has length
$\ge y(x)-y(f)$. Furthermore it shows that
an outer vertex $x$ has  $evenlength(x)=y(x)-y(f)=|P(x)|$.

{\em oddlevel}$(x)$, the length of a shortest
odd alternating
path from a free vertex to $x$, has a similar characterization, e.g.,
any odd alternating $fx$-path has length
$\ge 1-y(x)-y(f)+\sum z(B)$, where the sum extends over blossoms
$B$ with base vertex $b$ and $x\in B-b$.

Finally \cite{V2} divides the edges of $H$ into {\em bridges} and {\em props}.
This is due to the fact that an edge $e$ of form $Q$ can 
trigger an initial blossom step, which can be followed
by blossom steps triggered by unmatched edges of $T'$.
$e$ is a bridge and the other triggers are props.
(In the precise blossom structure stated above, $e$ is the $Q$ edge and the
prop triggers are in $T'$.)
\end{document}